\documentclass[11pt,reqno]{amsart}

\usepackage{amsmath,amsfonts,amsbsy}
\numberwithin{equation}{section}

\usepackage{enumitem}
\setlist{leftmargin=*}

\usepackage[a4paper]{geometry}
\newgeometry{margin=3cm}

\usepackage{xy}
\xyoption{matrix} \xyoption{arrow} \xyoption{arc} \xyoption{color}

\usepackage{tikz}
\usetikzlibrary{arrows}
\usetikzlibrary{intersections}
\usetikzlibrary{calc}

\usepackage{graphicx}
\usepackage{caption}
\usepackage{subcaption}

\usepackage{thmtools}

\declaretheoremstyle[
spaceabove=\medskipamount, spacebelow=\medskipamount,
headfont=\bfseries,
notefont=\bfseries\boldmath, notebraces={(}{)},
bodyfont=\itshape,
postheadspace=.5em
]{cursive}
\declaretheorem[style=cursive,name=Theorem,numberwithin=section]{theorem}
\declaretheorem[style=cursive,name=Lemma,numberlike=theorem]{lemma}

\declaretheorem[style=cursive,name=Proposition,numberlike=theorem]{proposition}

\declaretheoremstyle[
spaceabove=\medskipamount, spacebelow=\medskipamount,
headfont=\bfseries,
notefont=\bfseries\boldmath, notebraces={(}{)},
bodyfont=\rmfamily,
postheadspace=.5em
]{upright}
\declaretheorem[style=upright,name=Definition,qed=$\Diamond$,numberlike=theorem]
{definition}
\declaretheorem[style=upright,name=Remark,qed=$\Diamond$,numberlike=theorem]
{remark}

\newcommand{\Id}{\mathbf{1}}
\newcommand{\bk}{\mathbf{k}}

\newcommand{\C}{\mathbb{C}}
\newcommand{\R}{\mathbb{R}}
\newcommand{\Z}{\mathbb{Z}}

\newcommand{\T}{\mathbb{T}}
\newcommand{\D}{\mathbb{D}}

\newcommand{\Hi}{\mathcal{H}}
\newcommand{\A}{\mathcal{A}}
\newcommand{\F}{\mathcal{F}}
\newcommand{\BH}{\mathcal{B}(\mathcal{H})}

\newcommand{\scal}[2]{\left\langle #1, #2 \right\rangle}
\newcommand{\bra}[1]{\left\langle #1 \right|}
\newcommand{\ket}[1]{\left| #1 \right\rangle}

\newcommand{\eu}{\mathrm{e}}
\newcommand{\iu}{\mathrm{i}}
\newcommand{\di}{\mathrm{d}}

\newcommand{\FKM}{\mathrm{FKM}}
\newcommand{\BZ}{\mathrm{BZ}}
\newcommand{\EBZ}{\mathrm{EBZ}}
\newcommand{\WZ}{S\sub{WZ}}

\newcommand{\ie}{{\sl i.\,e.\ }}
\newcommand{\eg}{{\sl e.\,g.\ }}

\newcommand{\sub}[1]{_{\mathrm{#1}}}
\newcommand{\set}[1]{ \left\{  #1 \right\}}

\DeclareMathOperator{\tr}{tr}
\DeclareMathOperator{\Tr}{Tr}
\DeclareMathOperator{\Ran}{Ran}
\DeclareMathOperator{\diag}{diag}

\makeatletter
\newcommand\footnoteref[1]{\protected@xdef\@thefnmark{\ref{#1}}\@footnotemark}
\makeatother

\title[Gauge-theoretic invariants for topological insulators]{Gauge-theoretic invariants for topological insulators: \\ A bridge between Berry, Wess--Zumino, and Fu--Kane--Mele}
\author{Domenico Monaco \and Cl\'ement Tauber}
\date{January 19th, 2017 -- arXiv version 2, accepted for publication in {\it Lett. Math. Phys.}} % UPDATE IN FUTURE VERSIONS

\usepackage[pdftex,%
colorlinks=true,linkcolor=blue,citecolor=red,%
]
{hyperref}

\begin{document}

\begin{abstract}
We establish a connection between two recently-proposed approaches to the understanding of the geometric origin of the Fu--Kane--Mele invariant $\FKM \in \Z_2$, arising in the context of $2$-dimensional time-reversal symmetric topological insulators. On the one hand, the $\Z_2$ invariant can be formulated in terms of the Berry connection and the Berry curvature of the Bloch bundle of occupied states over the Brillouin torus. On the other, using techniques from the theory of bundle gerbes it is possible to provide an expression for $\FKM$ containing the square root of the Wess--Zumino amplitude for a certain $U(N)$-valued field over the Brillouin torus. 

We link the two formulas by showing directly the equality between the above mentioned Wess--Zumino amplitude and the Berry phase, as well as between their square roots. An essential tool of independent interest is an equivariant version of the adjoint Polyakov--Wiegmann formula for fields $\T^2 \to U(N)$, of which we provide a proof employing only basic homotopy theory and circumventing the language of bundle gerbes.

\medskip

\end{abstract}

\maketitle

\tableofcontents

\section{Introduction}

Introduced by Fu, Kane and Mele to characterize $2$-dimensional time-reversal symmetric topological insulators \cite{KaneMele05,FuKane06}, the eponym invariant $ \FKM \in \mathbb Z_2$ has now been investigated for one decade, especially in regard to its geometric interpretation. Despite the fact that Fu and Kane immediately suggested that it captures the existence ($\FKM=0$) or not ($\FKM=1$) of a set of compatible Kramer's pairs over the whole Brillouin torus $\BZ \cong \mathbb T^2$, the interpretation of FKM as an obstruction to define a global, smooth, periodic, and time-reversal symmetric Bloch frame for a general family of projectors $P(\bk)$ for $\bk \in \BZ$ (which, in applications to topological insulators, projects on the space of occupied Bloch states at fixed crystal momentum $\bk$) was mathematically established only recently in \cite{FiorenzaMonacoPanati16}. 

\begin{figure}[htb]
\centering
\pgfdeclarelayer{bg}
\pgfsetlayers{bg,main}
\begin{tikzpicture}[>=latex]
\node[draw, rectangle, inner sep=0, minimum height=4cm, minimum width=4cm] (BZ) at (0,0) {};
\draw[->] ($(BZ.west)!-0.1!(BZ.east)$) -- ($(BZ.west)!1.2!(BZ.east)$) node[right] {$k_1$};
\draw[->] ($(BZ.south)!-0.1!(BZ.north)$) -- ($(BZ.south)!1.2!(BZ.north)$) node[right] {$k_2$};
\coordinate (BZ+C) at ($(BZ.center)!1/2!(BZ.east)$);
\begin{pgfonlayer}{bg}
\node[fill=red!30, rectangle, inner sep=0, minimum height=4cm,minimum width=2cm] (BZ+) at (BZ+C) {};
\node[red] at ($(BZ+.south)!3/4!(BZ+.north)$) {$\text{EBZ}$};
\end{pgfonlayer}
\draw[very thick, red] (BZ.south) -- (BZ.north);
\draw[very thick, red] (BZ.south east) --  (BZ.north east);
\node[left,red] at ($(BZ.center)!0.4!(BZ.north)$) {$\mathbb T_0$} ;
\node[right,red] at ($(BZ.east)!0.4!(BZ.north east)$) {$\mathbb T_\pi$} ;
\node[below left] at (BZ.south) {$-\pi$};
\node[below right] at (BZ.east) {$\pi$};
\node[above left] at (BZ.north) {$\pi$};
\node[below left] at (BZ.west) {$-\pi$};
\node[fill=black,circle, inner sep=2pt] at (BZ.center) {};
\node[fill=black,circle, inner sep=2pt] at (BZ.north) {};
\node[fill=black,circle, inner sep=2pt] at (BZ.east) {};
\node[fill=black,circle, inner sep=2pt] at (BZ.north east) {};
\end{tikzpicture}
\caption{The Brillouin zone is periodic in $k_1$ and $k_2$, the effective Brillouin zone ($\EBZ$) is half of it and its boundaries are the two loops $\mathbb T_0$ and $\mathbb T_\pi$. The black dots are the four time-reversal invariant momenta, namely the inequivalent points invariant under $\bk \mapsto -\bk$. \label{fig:EBZ}}
\end{figure}
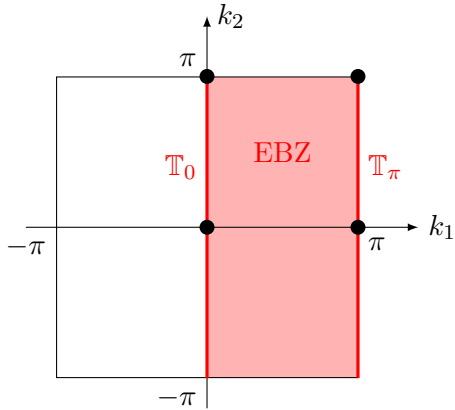

Besides, the question regarding the explicit computation of $\FKM$ for a given model has been also intensively studied. Notice how any formula for the topological invariant should take into account the symmetries of the physical system, namely periodicity and time-reversal symmetry: the former allows to focus one's attention in $\bk$-space to the Brillouin zone $\BZ$, while the latter further reduces the relevant points to consider to the \emph{effective Brillouin zone} $\EBZ$. Topologically, the Brillouin zone is a $2$-torus, while the effective Brillouin zone can be regarded a cylinder whose boundary is constituted by the two 1-dimensional tori (or loops) $\T_0$ and $\T_\pi$ (see Figure~\ref{fig:EBZ}). 

The initial definition of $\FKM$ in terms of a Pfaffian formula \cite{FuKane06}, requiring the evaluation of certain quantities only over the four time-reversal invariant momenta of the Brillouin zone (see Figure~\ref{fig:EBZ}), has the advantage of being compact and easy to compute, but avoids somehow to take into account the geometric framework behind. In particular, the Pfaffian formula is not well suited for generalizations, such as for example periodically driven or disordered systems. A very different situation occurs for systems with broken time-reversal symmetry, like quantum Hall systems (see \cite{Graf review} and references therein) and Chern insulators \cite{Haldane88, Bestwick et al 2015, Chang et al 2015}, where the topological invariant was early recognized to be the (first) Chern number. The latter can be computed through a ``local'' formula, namely integrating the Berry curvature, defined directly in terms of the family of projectors $P(\bk)$ as $\mathcal{F} = - \iu \Tr \left\{ P (\di P)^2 \right\}$, over the Brillouin torus (compare \eqref{C_1(P)}), and has moreover an interpretation again as a topological obstruction to the existence of a global, smooth, and periodic Bloch frame \cite{Nenciu91, Panati07, Monaco16}. To overcome the difficulties of the Pfaffian formula for $\FKM \in \Z_2$ and clarify its geometric origins, other approaches have been proposed to compute the Fu--Kane--Mele invariant \cite{Prodan09, Prodan11, GrafPorta13, Schulz-Baldes15, DeNittisGomi15}.

From the geometric obstruction formalism mentioned above, a different formula for $\FKM$ can be derived \cite{FuKane06,CorneanMonacoTeufel16,Monaco16}, namely
\begin{equation} \label{eq:def_delta}
\FKM = \delta, \quad \text{where}  \quad \delta = \frac{1}{2\pi} \left(\oint_{\T_\pi} \A - \oint_{\T_0} \A \right) - \frac{1}{2 \pi} \int_{\EBZ} \F \mod 2, 
\end{equation}
where, in contrast to the Chern number formula \eqref{C_1(P)}, the Berry curvature $\F$ is integrated only over the effective Brillouin zone $\EBZ$. The extra ``boundary'' terms involve the Berry connection $\A$, computed with respect to a time-reversal symmetric frame on $\T_{0/\pi}$; the fact that the expression is well-defined mod $2$ essentially accounts for the dependence of $\A$ on a gauge preserving time-reversal symmetry of Bloch frames. 

In a completely independent way, the Fu--Kane--Mele invariant was recently computed in terms of the uniquely-defined square root of a Wess--Zumino amplitude 
\cite{Lyon15,Gawedzki15}
\begin{equation}\label{eq:WZ_sr}
(-1)^{\FKM} = \sqrt{\eu^{\iu \WZ[U_P]}} \qquad \text{where} \qquad U_P(\bk) = \Id - 2 P(\bk) \in U(N).  
\end{equation}
This topological term was defined in the context of quantum field theory as a holonomy over a bundle gerbe \cite{GawedzkiReis02}, a powerful but somewhat heavy formalism that allows to implement time-reversal invariance properly on $\WZ$ so that the square root of the amplitude is well-defined and ends up coinciding with the Pfaffian formula when computed for the field $U_P$. Moreover, this framework was also used to consider 3-dimensional invariants \cite{Gawedzki15}, which as a side effect gives the effective reformulation%
\footnote{This formula is not explicitly written in \cite{Gawedzki15}, but is analogous to formula (II.53) in \cite{Gawedzki15} and its derivation is explained in words at the end of Section~II.D, page 22 there.}%
\begin{equation}\label{eq:def_K}
(-1)^{\FKM} = \mathcal{K}, \quad \text{where} \quad \mathcal{K} :=\dfrac{\sqrt{\exp\left(\iu \, \WZ[\phi_\pi] \right)}}{\sqrt{\exp\left(\iu \, \WZ[\phi_0] \right)}} \, \exp \left( \frac{\iu}{24 \pi} \int_{S^1\times\EBZ} \Tr \left\{ (\Phi^{-1} \di \Phi)^{3} \right\} \right) 
\end{equation}
where $\Phi(t,\mathbf k) := \exp \left( 2 \pi \iu t P(\mathbf k) \right) \in U(N)$ for $(t, \mathbf k) \in S^1 \times \BZ$ and $\phi_a(t,k) := \Phi(t,a,k)$ for $(t, k) \in S^1 \times \T$ and $a \in \set{0, \pi}$.  Additionally, when considering the definition of Wess--Zumino amplitude in terms of extension of fields (see Definition~\ref{def:WZ} below), this formula also gives a concrete expression for the formulation by Moore and Balents of the Fu--Kane--Mele invariant as an extension of the Chern number formula when restricted to the effective Brillouin zone \cite{MooreBalents07}.

Noticing that the square roots of the Wess--Zumino amplitudes appearing in \eqref{eq:def_K} are computed from restrictions of the family of projectors to the boundaries $\T_{0/\pi}$ of the effective Brillouin zone, and that moreover the integral over $S^1\times \EBZ$ on the right-hand side can be reduced to the one over $\EBZ$ of the Berry curvature $\mathcal F$ (see \eqref{eqn:PhiBerry} below), the formula above bears a strong similarity with a multiplicative version of the invariant $\delta$ defined in \eqref{eq:def_delta}. The aim of this paper is then to deduce expression \eqref{eq:def_K} for $\mathcal K$ directly from the expression \eqref{eq:def_delta} for $\delta$, using only differential calculs techniques and Witten's original definition of Wess--Zumino amplitudes in terms of fields extension \cite{Witten84}. We will do so in Theorem~\ref{thm:FKM}. In particular, this establishes independently the identity in \eqref{eq:def_K} and shows that the two approaches to compute $\FKM$ are equivalent by means of elementary tools from topology and differential geometry, without referring to the Pfaffian formula and circumventing the gerbe formalism. 

Our proof of \eqref{eq:def_K} is reduced to the following identity, which constitutes the main result of this paper (compare Theorem~\ref{thm:Main}, part~\ref{item:YesTRS}):
\begin{equation} 
\sqrt{\exp\left(\iu \, \WZ[\phi_a] \right)} = \sqrt{\exp \left( - \iu \oint_{\T_{a}} \A \right)}, \quad a \in \set{0, \pi}
\end{equation}
(the square root on the right-hand side is understood in the sense provided by Definition~\ref{eqn:SQRTBerryPhase} below). The above relation, together with the analogue one which holds before taking the square root when time-reversal invariance is broken (compare Theorem~\ref{thm:Main}, part~\ref{item:NoTRS}), is a result of independent interest, providing a simple interpretation of Wess--Zumino amplitudes for fields restricted to some loops in the Brillouin torus.

The paper is organized as follows. In Section \ref{sec:Definitions} we define properly all the quantities mentioned previously and needed to state the main results of the paper. Then Section \ref{sec:Homotopy} is dedicated to the proof of some homotopy invariance properties of Wess--Zumino amplitudes: as an interesting byproduct, we are able to show the validity of the (equivariant) adjoint Polyakov--Wiegmann formula, computing the (square root of the) Wess--Zumino amplitude of a product of fields of the form $g h g^{-1}$, in the specific case of $U(N)$-valued fields defined on $\Sigma = \T^2$ (Theorem~\ref{thm:APW}). This is used in Section \ref{sec:Proofs} to prove the main Theorem~\ref{thm:Main}. Finally Section \ref{sec:Conclusions} concludes with some possible generalizations and perspectives. The Appendix collects some results concerning the homotopy classes of $U(N)$-valued maps defined on the circle or on the $2$-torus, which are used several times throughout the paper.

\bigskip

\noindent \textsc{Acknowledgments.} The authors thank D.~Fiorenza and K.~Gaw\c{e}dzki for fruitful discussions and their useful comments about this work. D.M. acknowledges financial support from the German Science Foundation (DFG) within the GRK 1838 ``Spectral theory and dynamics of quantum systems''. The work of C.T. was supported by the PRIN project ``Mathematical problems in kinetic theory and applications'' (prot. 2012AZS52J).

\section{Definitions and main results} \label{sec:Definitions}

In what follows we define the different objects that have appeared in the Introduction, and state our main results. All are based on the same input describing the physical model, namely  a family $P(\bk)$, $\bk \in \R^2$, of rank-$m$ projectors in $\BH$, $\Hi = \C^N$, which is smooth and $(2 \pi \Z)^2$-periodic in $\bk$. Thus $P(\bk)$ is effectively defined for $\bk \in \T^2 := \R^2 / (2 \pi \Z)^2$: in the following, we will identify $\T = \R / 2 \pi \Z$ with the interval $[-\pi, \pi]$ with endpoints identified. At times, we may also require the family to be time-reversal symmetric with odd time-reversal symmetry, \ie
\begin{equation}\label{eq:TRI_P}
 P(-\bk) = \theta \, P(\bk) \, \theta^{-1} 
\end{equation}
for some antiunitary operator $\theta \colon \Hi \to \Hi$ such that $\theta^2 = - \Id$. The presence of such operator immediately implies that the dimension of $\Hi$ is even, $N = 2M$, because 
\begin{equation} \label{eqn:SymplecticForm}
(\varphi_1, \varphi_2) := \scal{\theta \, \varphi_1}{\varphi_2}, \quad \varphi_1, \, \varphi_2 \in \Hi
\end{equation}
defines a symplectic form over $\Hi$. Moreover, without loss of generality we may assume \cite{Hua44} that $\theta$ is of the form $\theta = J \, K$, where $J$ is the symplectic matrix 
\begin{equation} \label{J}
J = \begin{pmatrix} 0 & 1 & & & \\ -1 & 0 & & & \\ & & \ddots & & & \\ & & & 0 & 1 \\ & & & -1 & 0 \end{pmatrix}
\end{equation} 
and $K$ denotes the complex conjugation operator (with respect to some basis in $\Hi$). By a similar argument, also the rank of the projectors $P(\bk)$ must be even, $m = 2n$: simply consider the restriction of the symplectic form on $\Hi$ defined above to the invariant subspace $\Ran P(\mathbf 0)$. 

In applications to condensed matter systems, such a family of projectors comes from a periodic and time-reversal symmetric Hamiltonian, and selects the occupied Bloch states at fixed crystal momentum $\bk$ below the Fermi energy, whenever the latter sits in a spectral gap:
\[ P(\bk) = \sum_{n}^{\mathrm{occ}} \ket{\psi_{n,\bk}} \bra{\psi_{n,\bk}} \]
(see \eg \cite{MonacoPanati15}).
 
\subsection{Bloch frames, Berry connection, Berry curvature}

We collect here some definitions of geometric objects which can be constructed out of a family of projectors $P(\bk)$ as above.

\begin{definition}[Bloch frame]
A \emph{Bloch frame} for the family of projectors $P(\bk)$ is a collection of vectors $\set{e_a(\bk)}_{1 \le a \le m}$ which give an orthonormal basis of the vector space $\Ran P(\bk) \subset \Hi$; equivalently, 
\[ P(\bk) = \sum_{a=1}^{m} \ket{e_a(\bk)} \bra{e_a(\bk)}. \]

A Bloch frame is called \emph{continuous}, \emph{smooth}, or \emph{periodic} if the corresponding functions $\bk \mapsto e_a(\bk)$ are continuous, smooth, or periodic for all $a \in \set{1, \ldots, m}$. A Bloch frame is called \emph{time-reversal symmetric} if
\[ e_{2j-1}(-\bk) = - \theta e_{2j}(\bk), \quad e_{2j}(-\bk) = \theta 
e_{2j-1}(\bk), \quad 1 \le j \le n = m/2. \qedhere \]
\end{definition}

The notion of Bloch frames allows to define the Berry connection, Berry phase and Berry curvature associated to $P(\bk)$, first introduced in \cite{Berry84}.

\begin{definition}[Berry connection]
Let $\set{e_a(\bk)}_{1 \le a \le m}$ be a Bloch frame for the family of projectors $P(\bk)$. The (\emph{abelian}) \emph{Berry connection} associated to the frame is the $1$-form
\begin{equation} \label{eqn:Berry}
\A := -\iu \sum_{a=1}^{m} \scal{e_a}{\di e_a}. \qedhere
\end{equation}
\end{definition}

Notice that if $\set{e_a'(\bk)}_{1\le a \le m}$ is another Bloch frame, then the Berry connection $\A'$ associated to this frame reads%
\footnote{\label{traces} We denote by $\tr\{\cdot\}$ the trace of an $m \times m$ matrix, and by $\Tr\{\cdot\}$ the trace of an operator on $\Hi$ (\ie of an $N \times N$ matrix).}%
\[ \A' = \A - \iu \tr \left\{ u^{-1} \di u \right\} , \]
if $u(\bk)$ denotes the change-of-basis matrix (also known as the \emph{gauge}) between $\set{e_a(\bk)}$ and $\set{e_a'(\bk)}$, that is, 
\[ [u(\bk)]_{ab} := \scal{e_a(\bk)}{e_b'(\bk)}. \]
The Berry connection is thus \emph{gauge-dependent}, \ie it depends on the choice of the Bloch frame. 

Observe also that, if $\T$ is a loop in $\T^2$ (say \eg the one where one of the coordinates $(k_1, k_2)$ stays constant), and if a continuous, periodic Bloch frame on $\T$ exists, then
\[ \frac{1}{2 \pi} \oint_\T \A' = \frac{1}{2 \pi} \oint_\T \A + \frac{1}{2 \pi \iu} \oint_\T \tr \left\{ u^{-1} \di u \right\}. \]
The second summand on the right-hand side of the above equality computes the winding number of the (periodic) map $\T \ni \bk \mapsto \det u(\bk) \in U(1)$ (see \eg Lemma~\ref{T->U(N)} in the Appendix), and is thus an integer. This leads to set the following

\begin{definition}[Berry phase] \label{eqn:BerryPhase}
The gauge-independent quantity 
\[ \exp \left( -\iu \oint_\T \A \right) \quad \in U(1) \]
is called the \emph{Berry phase}.
\end{definition}

When $\set{e_a(\bk)}$ is a time-reversal symmetric Bloch frame, we may ask how the integral of $\A$ over a loop changes when choosing a different gauge which preserves this symmetry. It is easily verified that the frame $\set{e_a'(\bk)}$ will be again time-reversal symmetric if the relative gauge $u(\bk)$ satisfies
\begin{equation} \label{TRSgauge} 
u(-\bk) = J^{-1} \, \overline{u(\bk)} \, J, 
\end{equation}
where $J$ is the $m \times m$ symplectic matrix as in \eqref{J}. If now the loop $\T$ is left invariant by the involution $\bk \mapsto -\bk$ (as is the case for example for the loops $\T_{0/\pi}$ in Figure~\ref{fig:EBZ}), then the above relation implies that the winding number of $\bk \mapsto \det u(\bk)$ is an \emph{even} integer (compare part~\ref{1d-Equiv} of Lemma~\ref{T->U(N)}). This yields that $\oint_\T \A$ is actually well-defined mod $4 \pi \Z$, and we can define

\begin{definition}[Square root of the Berry phase] \label{eqn:SQRTBerryPhase}
The \emph{square root of the Berry phase} is given by
\[ \sqrt{\exp \left( -\iu \oint_\T \A \right)} := \exp \left( -\frac{\iu}{2} \oint_\T \A \right) \quad \in U(1) \]
where $\A$ is computed with respect to a smooth, periodic and time-reversal symmetric Bloch frame. This quantity is invariant under changes of gauge which preserve time-reversal symmetry.
\end{definition}

The Berry phase and its square root will play a prominent role in our main results. 

\begin{definition}[Berry curvature and Chern number]
The \emph{Berry curvature} associated to the family of projectors $P(\bk)$ is given by%
\footnote{The product of differential forms is always the wedge product. Thus $\omega^n := \omega \wedge \stackrel{(n \text{ times})}{\cdots} \wedge\, \omega$.

For the definition of the trace $\Tr \{\cdot\}$, compare the footnote on page \pageref{traces}.}%
\[ \F := - \iu \Tr \left\{ P (\di P)^2 \right\}. \]

The \emph{Chern number} of $P(\bk)$ is the integer
\begin{equation} \label{C_1(P)}
C_1(P) := \frac{1}{2\pi} \int_{\T^2} \F \quad \in \Z. \qedhere
\end{equation}
\end{definition}

Since it is expressed directly in terms of the projectors, the Berry curvature (and hence the Chern number) is a gauge-invariant quantity. Moreover, a long but straightforward computation shows that the Berry curvature $2$-form is the differential of the Berry connection $1$-form, that is, $\F = \di \A$.

\begin{remark}[Geometric interpretation: Bloch bundle] \label{BlochBundle}
To each smooth and periodic family of projectors, one can associate (via the Serre--Swan construction) a vector bundle $\mathcal{E}$ over the Brillouin $2$-torus $\T^2$, called the \emph{Bloch bundle}. Roughly speaking, $\mathcal{E}$ is determined by requiring that its fiber over the point $\bk \in \T^2$ be the $m$-dimensional vector space $\Ran P(\bk)$ -- see \cite{Panati07, MonacoPanati15} for details.

All the terminology employed in this Section has been borrowed indeed from the language of vector bundles. Bloch frames are nothing but trivializing frames for the Bloch bundle: as such, they are in general only defined locally on $\T^2$. The Berry connection is the trace of the connection $1$-form associated to the Grassmann connection, induced by the trivial connection via the obvious inclusion $\mathcal{E} \hookrightarrow \T^2 \times \Hi$. In the same way, the Berry curvature is the trace of the curvature $2$-form for the Grassmann connection. The gauge invariance of the Berry phase can be also understood from the fact that it equals the determinant of the holonomy associated to the Berry connection along the loop \cite[Prop.~4.3]{CorneanMonacoTeufel16} (see also \cite{Simon83, KarpMansouriRno00, FiorenzaSatiSchreiber15}).

The Chern number is a topological invariant associated to the Bloch bundle $\mathcal{E}$, and it characterizes its isomorphim class as a bundle over the $2$-dimensional torus \cite{Panati07}. In particular, it vanishes exactly when $\mathcal{E}$ is trivial (\ie isomorphic to the product bundle $\T^2 \times \C^m$) or, equivalently, when a smooth and periodic Bloch frame for $P(\bk)$ exists on the whole $\T^2$ \cite{Panati07}. The fact that the formula \eqref{C_1(P)} computes an integer is a non-trivial statement, which can be proved for example using obstruction theory \cite{Monaco16}. In applications to solid-state physics, the Chern number appears as a theoretical explanation for the quantization of the Hall conductance in quantum Hall systems, see \cite{Graf review} and references therein.
\end{remark}

As a last observation, notice that  the Chern number vanishes for time-reversal symmetric families of projectors \cite{Panati07, MonacoPanati15}. Indeed, the Berry curvature can be written as
\[ \F = \Omega(\bk) \, \di k_1 \wedge \di k_2, \quad \text{where} \quad \Omega(\bk) := - \iu \Tr \left\{ P(\bk) \left[ \partial_{k_1} P(\bk), \partial_{k_2} P(\bk) \right] \right\}. \]
If $P(\bk)$ satisfies \eqref{eq:TRI_P}, then $\Omega(\bk) = - \Omega(-\bk)$ is odd, and hence integrates to zero over $\T^2$.

\subsection{Wess--Zumino amplitudes and their square roots}

We abandon momentarily the language of families of projectors to discuss the second main character in the results of the present paper, namely the \emph{Wess--Zumino amplitude}. Here we consider Witten's original definition of the Wess--Zumino action \cite{Witten84}, namely 

\begin{definition}[Wess--Zumino action] \label{def:WZ}
A \emph{field} is a smooth map $g : \Sigma \rightarrow G$, where $\Sigma$ a $2$-dimensional compact and closed surface and $G$ is a (compact, matrix) Lie group. An \emph{extension} of a field is a map $\widetilde{g} \colon \widetilde{\Sigma} \to G$ where $\widetilde{\Sigma}$ is a $3$-dimensional manifold with boundary  $\partial \widetilde \Sigma = \Sigma$ such that the restriction $\widetilde g |_{\partial \widetilde \Sigma} = g$. 

The \emph{Wess--Zumino action} of a field $g$ is then defined by
\begin{equation} \label{chi}
\WZ[g] := \int_{\widetilde \Sigma} {\widetilde g}^* \chi, \qquad \text{where} \qquad \widetilde{g}^* \chi = \frac{1}{12\pi} \Tr \left\{ (\widetilde{g}^{-1} \di \widetilde{g})^{3} \right\}.
\end{equation}
Thus $\widetilde{g}^* \chi$ is the pullback via $\widetilde{g}$ of the $3$-form $\chi$ on the group $G$.
\end{definition}

Let us point out two obvious problems that this definition of $\WZ[g]$ poses:
\begin{enumerate}
 \item it requires the existence of an extension for the given field $g$;
 \item the value it computes could \textit{a priori} depend on the choice of the extension.
\end{enumerate}

When $G$ is simply connected (like for example $G= SU(N)$), the existence of an extension for any map $g \colon \Sigma \to G$ to a $3$-manifold with boundary $\Sigma$ is guaranteed by the fact that $G$ is $2$-connected, as $\pi_2(G) = 0$ for any compact Lie group. However, in the applications we have in mind, the surface $\Sigma$ will always be a 2-dimensional torus $\T^2$, and the Lie group will always be the non-simply-connected unitary group $U(N)$. In this case, the existence of an extension is guaranteed if the winding number of $g$ along one direction of $\Sigma = \T\times \T$ vanishes, as in this case the field can be extended to the solid torus, where the corresponding circle $\T$ is ``filled'' to a unit disk%
\footnote{Indeed, the vanishing of the winding number along $\T$ implies that the restriction of the field to $\T$ is homotopic to a constant map (Lemma~\ref{T->U(N)}); besides, an homotopy $F \colon \T \times [0,1] \to U(N)$ from a constant map $f_0$ to a map $f_1$ provides an extension of $f_1$ to the unit disk by setting $f(r \, \eu^{\iu k}) := f_r(k)$, $(k, r) \in \T \times [0,1]$.}%
. This will always be the case for the fields considered below. Besides, if $g$ has non-vanishing winding numbers in both directions, the previous definition is still meaningful using the invariance of the Wess--Zumino action under diffeomorphisms of $\Sigma$ \cite{Gawedzki88,GawedzkiReis02}, since a reparametrization of the torus always allows to ``unwind'' the field in one direction. 

As for the second problem with the definition of $\WZ$, one notices that when $G=U(N)$ the Wess--Zumino action is properly normalized so that for two different extensions $(\widetilde \Sigma_1,\widetilde g_1)$ and $(\widetilde \Sigma_2,\widetilde g_2)$ the difference between the two computations%
\footnote{This fact depends on the proper normalization of $\chi$ as the unique generator of the third cohomology group $H^3(U(N);\Z)$, see \cite{ChevalleyEilenberg, BottSeeley78, Gawedzki99}. For more general Lie groups the Wess--Zumino action is not ambiguous only when $\chi$ is multiplied by specific levels $k \in \mathbb Z$.} 
\[
\WZ^{(1)}[g] - \WZ^{(2)}[g] \in 2 \pi \mathbb Z
\]
so that the corresponding \emph{Wess--Zumino amplitude} 
\[ \exp(\iu \WZ[g] ) \quad \in U(1) \]
is well-defined. 

As a simple example of Wess--Zumino computation, and to see a first relation with the previous Subsection, we provide the following

\begin{proposition} For a smooth family of projectors $P(\bk)$ with $\bk \in \mathbb T^2$, consider the field $U_P(\bk) := \Id - 2 P(\bk) \in U(N)$. Then
\begin{equation}\label{eq:WZ_UP}
\eu^{\iu \WZ[U_P]} = (-1)^{C_1(P)}
\end{equation}
where $C_1(P)$ is the Chern number of $P(\bk)$ defined in \eqref{C_1(P)}.
\end{proposition}
\begin{proof}
Consider the following extension: $\widetilde \Sigma := [0,1]\times \mathbb T^2$ and
\[ \widetilde U_P(t,\bk) := \eu^{\iu \pi t P(\bk)}  = \eu^{\iu \pi t}P(\bk) + \Id - P(\bk)\]
where the last equality comes from the spectral decomposition of $P(\bk)$. Moreover $\widetilde U_P(1,\bk) = U_P(\bk)$ and $\widetilde U_P(0,\bk) = \Id$, so that the boundary $\{0\} \times \mathbb T^2$ gives a trivial contribution in the computation of the Wess--Zumino action. Thus a direct computation shows 
\[\widetilde U_P^{-1} \di \widetilde U_P = \iu \pi \di t P + (\eu^{\iu \pi t}-1) \di P + 2 (1-\cos(\pi t )) P \di P \]
where we have used that $(\Id-P)P=0$ since $P$ is a projector. Hence, after some algebra
\begin{equation} \label{eqn:UP}
\Tr\big((\widetilde U_P^{-1} \di \widetilde U_P)^3\big) =  6 \iu \pi (\cos(\pi t)-1) \di t  \Tr\left\{ P (\di P)^2 \right\}
\end{equation}
where we have used that $P(\di P) P=0$, so that when integrating over $t$ we get
\[ \WZ[U_P] = -\frac{\iu}{2}\int_{\mathbb T^2} \Tr\left\{ P (\di P)^2 \right\} 
=\pi\, C_1(P) \] 
which concludes the proof by taking the corresponding amplitude.
\end{proof} 

The analogue of the above statement in the context of a time-reversal symmetric topological insulators is \eqref{eq:WZ_sr}. As was already noticed, when \eqref{eq:TRI_P} is satisfied the Berry curvature is odd and the Chern number vanishes, so that the Wess--Zumino amplitude in \eqref{eq:WZ_UP} is always 1. This shows in particular that its square root, which appears in \eqref{eq:WZ_sr}, belongs to $\set{\pm 1} = \mathbb Z_2$. However the explicit computation of the square root for the Wess--Zumino amplitude of any time-reversal invariant map requires the technology of (Hermitian line) bundle gerbe with unitary connection, and was already deeply investigated in \cite{Lyon15,Gawedzki15}.

Here we would like to circumvent this approach and define the square root of Wess--Zumino amplitudes \`a-la-Witten, that is, via field extensions, at least for $U(N)$-valued fields. To this end, recall that by a time-reversal symmetry we mean an antiunitary operator $\theta \colon \Hi \to \Hi$ on the Hilbert space $\Hi = \C^{2M}$ such that $\theta^2 = - \Id$. The induced adjoint action of the time-reversal symmetry operator on the unitaries over $\Hi$ will be denoted by
\begin{equation} \label{Theta}
\Theta(g) := \theta \, g \, \theta^{-1}, \quad g \in U(2M).
\end{equation}

\begin{definition}[Equivariant fields and extensions] \label{def:EquivField}
A field $g \colon \Sigma \to U(2M)$ on the $2d$ compact surface $\Sigma$ will be called \emph{$\Z_2$-equivariant} (or simply \emph{equivariant}) if there exists an involution $\vartheta \colon \Sigma \to \Sigma$ such that
\[ \Theta \circ g = g \circ \vartheta. \]

By an \emph{equivariant extension} of the field $g \colon \Sigma \to U(2M)$ we mean a map $\widetilde{g} \colon \widetilde{\Sigma} \to U(2M)$ from a $3d$ manifold $\widetilde{\Sigma}$ with boundary $\partial \widetilde{\Sigma} = \Sigma$ such that
\begin{itemize}
 \item $\widetilde{g} \big|_{\partial \widetilde{\Sigma}} \equiv g$, and
 \item there exists an involution $\widetilde{\vartheta}$ on $\widetilde{\Sigma}$ such that $\widetilde{\vartheta} \big|_{\partial \widetilde{\Sigma}} \equiv \vartheta$ and $\Theta \circ \widetilde{g} = \widetilde{g} \circ \widetilde{\vartheta}$. \qedhere
\end{itemize}
\end{definition}

\medskip

For equivariant fields, a finer notion than the Wess--Zumino amplitude can be introduced.

\begin{definition}[Square root of the Wess--Zumino amplitude]
For an equivariant field $g \colon \Sigma \to U(2M)$, the \emph{square root of the Wess--Zumino amplitude} is defined as 
\[ \sqrt{\exp(\iu \, \WZ[g])} := \exp \left( \frac{\iu}{2} \int_{\widetilde{\Sigma}} \widetilde{g}^* \chi \right) \]
where $\widetilde{g}$ is an equivariant extension of $g$ and $\chi$ is the $3$-form over $U(2M)$ appearing in \eqref{chi}. 
\end{definition}

Notice that equivariance of the extension makes the quantity $\int_{\widetilde{\Sigma}} \widetilde{g}^* \chi$ well defined $\bmod \: 4 \pi \Z$ rather than $\bmod \: 2 \pi \Z$, see \cite{Lyon15} and \cite[Prop.~1]{Gawedzki15}. The latter references show also that there exists indeed an equivariant extension for any equivariant field $g \colon \T^2 \to U(2M)$.

\subsection{Main results}

We are finally able to state the main results. These concern the evaluation of Wess--Zumino amplitudes for very specific fields that are defined starting from a smooth, periodic, and possibly time-reversal symmetric family of projectors $P(\bk)$, $\bk \in \R^2$. These fields appeared in \eqref{eq:def_K} and read
\begin{equation} \label{phia}
\phi_a(t,k) := \exp( 2 \pi \iu t P(a,k)), \qquad (t,k) \in S^1 \times \mathbb T_a, \qquad \text{for } a \in \{0,\pi\} 
\end{equation}
where $S^1 = \R / \Z$ and $\T_{0/\pi}$ are the boundaries of the effective Brillouin zone $\EBZ$ (see Figure~\ref{fig:EBZ}). The Wess--Zumino amplitude of $\phi_a$ will be expressed in terms of the Berry phase of the projector along the loop $\T_a$, both in presence and in absence of time-reversal symmetry.

\begin{theorem} \label{thm:Main}
Let $P(\bk)$, $\bk \in \R^2$, be a smooth and periodic family of projectors on $\Hi \simeq \C^N$. Define $\phi_a \colon S^1 \times \T_a \to U(N)$ as in \eqref{phia} for $a \in \set{0,\pi}$.
\begin{enumerate}
 \item \label{item:NoTRS} The Wess--Zumino amplitude of the field $\phi_a$ equals the Berry phase of the projectors along $\T_a$, \ie
 \[ \exp \left( \iu \WZ[\phi_a] \right) = \exp \left( - \iu \oint_{\T_a} \A \right). \]
 \item \label{item:YesTRS} If moreover the family of projectors is time-reversal symmetric, then the square root of the Wess--Zumino amplitude of the field $\phi_a$ equals the square root of the Berry phase of the projectors along $\T_a$, \ie
 \[ \sqrt{\exp \left( \iu \WZ[\phi_a] \right)} = \sqrt{\exp \left( - \iu \oint_{\T_a} \A \right)}. \]
\end{enumerate}
\end{theorem}

Notice that part of the statement entails the well-posedness of the (square root of the) Berry phase, that is, of the existence of a smooth, periodic (and time-reversal symmetric) Bloch frame for $P(\bk)$ along $\T_a$. Both statements in Theorem~\ref{thm:Main} can be seen as incarnations of a ``dimensional reduction'', where an intrinsically $2$-dimensional object like the Wess--Zumino amplitude of the specific field $\phi_a$ can be computed by an integration over a $1$-dimensional loop, rather than by a $3$-dimensional extension. The proof of Theorem~\ref{thm:Main} can be found in Section~\ref{sec:Proofs} (compare Theorem~\ref{thm:WZ=Berry}).

As an application of the above result to the context of topological insulators, we are able to show directly the equality between the two formulations \eqref{eq:def_delta} and \eqref{eq:def_K} for the Fu--Kane--Mele invariant $\FKM \in \Z_2$. The following statement can be seen as an alternative proof of \eqref{eq:def_K}, which avoids using the techniques of bundle gerbes adopted in \cite{Gawedzki15}.

\begin{theorem} \label{thm:FKM}
For a smooth, periodic, and time-reversal symmetric family of projectors $P(\bk)$, $\bk \in \R^2$, let $\delta$ be defined as in \eqref{eq:def_delta} and $\mathcal{K}$ be defined as in \eqref{eq:def_K}. Then
\[ \mathcal{K} = (-1)^\delta. \]
\end{theorem}
\begin{proof}
Since the expression computing $\delta$ in \eqref{eq:def_delta} is a well-defined integer mod $2$, we can compute
\begin{align*}
(-1)^\delta & = (\eu^{-\iu \pi})^{\frac{1}{2\pi} \left(\oint_{\T_\pi} \A - \oint_{\T_0} \A \right) - \frac{1}{2 \pi} \int_{\EBZ} \F} = \dfrac{\exp \left( -\frac{\iu}{2} \oint_{\T_\pi} \A \right)}{\exp -\left( \frac{\iu}{2} \oint_{\T_0} \A \right)} \, {\exp \left( \frac{\iu}{2} \int_{\EBZ} \F \right)} \\
& = \frac{\sqrt{\exp \left( - \iu \oint_{\T_\pi} \A \right)}}{\sqrt{\exp \left( - \iu \oint_{\T_0} \A \right)}} \, \exp \left( \frac{1}{2} \int_{\EBZ} \Tr \left\{ P (\di P)^2 \right\} \right).
\end{align*}
We compare the above expression with formula \eqref{eq:def_K} for $\mathcal{K}$. Theorem~\ref{thm:Main} gives the equality between the ratio on the right-hand side of the above equation and the one appearing in \eqref{eq:def_K}. The two exponential terms can instead be compared by noticing that 
\begin{equation} \label{eqn:PhiBerry}
\int_{[0,1]\times \EBZ}\Tr \left\{ (\Phi^{-1} \di \Phi)^{3} \right\} = - 12 \pi \iu  \int_{\EBZ} \Tr \left\{ P (\di P)^{2} \right\}
\end{equation}
for $\Phi(t,\bk) := \exp(2\pi \iu t P(\bk))$, $(t,\bk) \in [0,1]\times \R^2$. The above identity follows from an algebraic computation similar to one performed in the proof of Proposition \ref{eq:WZ_UP} (compare $\widetilde U_P$ and $\Phi$): replacing $t$ by $2t$ in \eqref{eqn:UP} and performing the integration over $t \in [0,1]$ leads to the result. This concludes the proof of the Theorem.
\end{proof}

\subsection{Factorization of fields} \label{sec:Factor}

To conclude this Section, we make some remarks on the stament of our main Theorem~\ref{thm:Main}, that will serve also as a motivation for the study of the Polyakov--Wiegmann formula in the next Section.

To streamline the notation, we consider a smooth, periodic (and possibly time-reversal symmetric) family of projectors $P(k)$, $k \in \T \equiv \R / 2 \pi \Z$, and set
\begin{equation} \label{eq:phi}
\phi(t, k) := \exp\left( 2 \pi \iu \, t \,P(k) \right), \quad (t,k) \in S^1 \times \T \simeq \T^2.
\end{equation}
When $P(k)$ is the restriction of a $2$-dimensional family of projectors to the two boundaries $\T_{0}$ and $\T_{\pi}$ of the effective Brillouin zone, we end up in the setting of the statement of Theorem~\ref{thm:Main}.

It is well known that out of the family of projectors $P(k)$ one can construct a rank-$m$ Hermitian vector bundle $\mathcal{E}$ over the circle $\T$, called the \emph{Bloch bundle} (compare Remark~\ref{BlochBundle}). As every vector bundle over the circle, it is isomorphic to the trivial bundle given by the Cartesian product $\T \times \C^m$. At the level of projectors, this implies the existence of a smooth, periodic family of unitary operators $W(k) \in U(N)$ such that
\begin{equation} \label{eqn:trivial}
P(k) = W(k) \, P(0) \, W(k)^*, \quad k \in \T.
\end{equation}
Moreover, $W(k)$ can be normalized so that $W(0) = \Id$ without loss of generality. When $P(k)$ is time-reversal symmetric, then $W(k)$ can be chosen to be time-reversal symmetric as well, meaning that $W(-k) = \theta \, W(k) \, \theta^{-1}$.

\begin{remark}[Parallel transport] \label{rmk:parallel}
Explicitly, such family of unitary operators can be constructed as follows \cite{CorneanMonacoTeufel16}. Define the \emph{parallel transport} unitary $T(k)$ as the solution to the operator-valued Cauchy problem
\[ \begin{cases}
\iu \, \partial_k T(k) = G(k) \, T(k), & G(k) := \iu [\partial_k P(k), P(k)] = 
G(k)^*, \\
T(0) = \Id. &
\end{cases} \]
The family of operators $T(k)$ is smooth (and possibly time-reversal symmetric), and satisfies the intertwining property
\[ P(k) = T(k) \, P(0) \, T(k)^*. \]
However, $T(k)$ is in general not periodic in $k$. Write $T(2\pi) = \eu^{2 \pi \iu M}$, $M=M^*$, via spectral decomposition. Then the family of operators
\[ W(k) := T(k) \eu^{- \iu k M} \]
satisfies all the required properties.
\end{remark}

The relation \eqref{eqn:trivial} yields  at once that
\begin{equation} \label{eqn:conjugate}
\phi(t,k) = W(k) \, \psi(t) \, W(k)^*, \quad \psi(t) := \exp(2 \pi \iu t P(0)).
\end{equation}
Thus, in order to compute the Wess--Zumino amplitude%
\footnote{Note that, since $\det \phi(k,t) = \det \psi(t)$ for all $k \in \T$, the map $\phi$ does not wind along the $k$-direction. Thus an extension of $\phi$ to the solid torus exists, and the corresponding Wess--Zumino action is well-defined through Definition \ref{def:WZ}. However, we won't need such an explicit extension as we will actually exploit the factorized structure of $\phi$ to compute its Wess--Zumino amplitude.} %
of the field $\phi$ and its square root, appearing in the statement of Theorem~\ref{thm:Main}, it is convenient to establish a general principle allowing to express the Wess--Zumino amplitude of a product of fields in terms of its factors. This is exactly what the Polyakov--Wiegmann formula accomplishes.

\section{Equivariant adjoint Polyakov--Wiegmann formula} \label{sec:Homotopy}

We collect in this Section several results concerning the \emph{Polyakov--Wiegmann formula} that allows to evaluate the Wess--Zumino amplitude of a product of fields. These constitute the main technical tools for the proof of Theorem \ref{thm:Main}, but are also of independent interest.

\subsection{Derivative of the Wess--Zumino action} 

We are first interested in the change of the Wess--Zumino action $\WZ[g]$ under homotopic deformation of the field $g \colon \Sigma \to G$, where $G$ is any compact Lie group. Recall that two continuous maps $f_0, \, f_1 \colon X \to Y$ between topological spaces are called \emph{homotopic} if there exists a continuous map $F \colon X \times [0,1] \to Y$ such that $F(x,0) = f_0(x)$ and $F(x,1) = f_1(x)$ for all $x \in X$. Hereinafter we denote $f_s(x) := F(x,s)$, $(x,s) \in X \times [0,1]$, if $F$ is an homotopy.

As a preliminary result we prove the following
\begin{proposition} \label{prop:HEP}
Let $g_0, \, g_1 \colon \Sigma \to G$ be homotopic fields. Assume that $g_0$ admits an extension $\widetilde{g}_0 \colon \widetilde{\Sigma} \to G$, in the sense of Definition~\ref{def:WZ}. Then there exists a \emph{smooth} homotopy $F : \Sigma \times [0,1] \to G$ between $g_0$ and $g_1$ which lifts to a \emph{smooth} homotopy of extensions $\widetilde{F} : \widetilde{\Sigma} \times [0,1] \to G$, that is, the maps $F$ and $\widetilde{F}$ are smooth with respect to $s \in [0,1]$ and moreover $\widetilde{f}_s \big|_{\partial\widetilde{\Sigma}} \equiv f_s$.
\end{proposition}
\begin{proof}
Pick any (continuous) homotopy $H$ between $g_0$ and $g_1$. By the homotopy extension principle \cite[Sec.~6.7]{DavisKirk01}, this lifts to a continuous homotopy $\widetilde{H} \colon \widetilde{\Sigma} \times [0,1] \to G$ such that $\widetilde{h}_s$ restricted to the boundary of $\widetilde{\Sigma}$ coincides with $h_s$ (in particular the restriction of $\widetilde{h}_1$ is $g_1$), and $\widetilde{h}_0$ is the given $\widetilde{g}_0$. The map $\widetilde{h}_1 \colon \widetilde{\Sigma} \to G$ is \textit{a priori} only continuous, but by the Whitney approximation theorem \cite[Thm.~10.21]{Lee03} it is homotopic to a smooth map relative to $\Sigma = \partial \widetilde{\Sigma}$, namely there exists an homotopy $\widetilde{H}' \colon \widetilde{\Sigma} \times [0,1] \to G$ with $\widetilde{h}'_0 \equiv \widetilde{h}_1$, $\widetilde{h}'_1 \colon \widetilde{\Sigma} \to G$ smooth and $\widetilde{h}'_s \big|_{\partial \widetilde{\Sigma}} \equiv g_1$ for all $s \in [0,1]$.

Denote by $\widetilde{H}''$ the concatenation of the two homotopies $\widetilde{H} \sharp \widetilde{H}'$, that is
\[ \widetilde{h}''_s := \begin{cases} \widetilde{h}_{2s} & \text{if } s \in [0,1/2], \\\widetilde{h}'_{2s-1} & \text{if } s \in [1/2,1], \end{cases} \quad s \in [0,1]. \]
Thus $\widetilde{H}'' \colon \widetilde{\Sigma} \times [0,1] \to G$ is a continuous homotopy between the smooth maps $\widetilde{h}_0 = \widetilde{g}_0$ and $\widetilde{h}'_1$. By standard approximation results \cite[Prop.~10.22]{Lee03}, $\widetilde{H}''$ can be replaced by a \emph{smooth} homotopy $\widetilde{F}$ between the two maps. The restriction $f_s$ of $\widetilde{f}_s$ to the boundary of $\widetilde{\Sigma}$ provides the desired smooth homotopy between $g_0$ and $g_1$.
\end{proof}

In view of the above result, hereinafter we will assume, whenever we speak of homotopic fields and homotopies of their extensions, that the homotopies depend smoothly also on the deformation parameter $s \in [0,1]$.

Consider now a smooth family  of fields $g_s : \Sigma \rightarrow G$. We use the smooth extension $\widetilde g_s : \widetilde \Sigma \to G$ provided by the above Proposition to define the corresponding action $S_{\rm WZ}[g_s]$. We want to compute 
\[ \dfrac{\di S_{\rm WZ}[g_s]}{\di s} = \dfrac{1}{12\pi} \int_{\widetilde \Sigma} \partial_s \Tr\big\{ (\widetilde g_s^{-1} \di \widetilde g_s )^3\big\}. 
\]
We have
\begin{align*}
\dfrac{\di S_{\rm WZ}[g_s]}{\di s}& \, = \dfrac{1}{4\pi} \int_{\widetilde \Sigma}  \Tr\Big\{ \big(- \widetilde g_s^{-1}(\partial_s \widetilde  g_s) \widetilde g_s^{-1} \di \widetilde g_s + \widetilde g_s^{-1} \di \partial_s \widetilde g_s \big) (\widetilde g_s^{-1} \di \widetilde g_s )^2\Big\}\cr
& \, = \dfrac{1}{4\pi} \int_{\widetilde \Sigma} \Tr\Big\{ \big(- (\partial_s \widetilde  g_s) \widetilde g_s^{-1} \di \widetilde g_s \, \widetilde g_s^{-1} +  \di \partial_s \widetilde g_s\,   \widetilde g_s^{-1}\big) (\di \widetilde g_s \,\widetilde g_s^{-1})^2\Big\}\cr
&\, = \dfrac{1}{4\pi} \int_{\widetilde \Sigma} \Tr \Big\{ \di (\partial_s \widetilde g_s \, \widetilde g_s^{-1})  (\di \widetilde g_s \,\widetilde g_s^{-1})^2 \Big\}
\end{align*}
where we have used the cyclicity of the trace in the second line. Then using the fact that $ \di \Tr (\di \widetilde g_s \,\widetilde g_s^{-1})^2 = 0$, we get
\begin{align*}
\dfrac{\di S_{\rm WZ}[g_s]}{\di s} &\, = \dfrac{1}{4\pi} \int_{\widetilde \Sigma} \di \Tr \Big\{  \partial_s \widetilde g_s \, \widetilde g_s^{-1}  (\di \widetilde g_s \,\widetilde g_s^{-1})^2 \Big\}\cr
&\, = \dfrac{1}{4\pi} \int_{\Sigma} \Tr \Big\{  \partial_s  g_s \,  g_s^{-1}  (\di g_s \, g_s^{-1})^2 \Big\}
\end{align*}
or equivalently by cyclicity of the trace
\begin{equation}\label{dSwz}
\dfrac{\di S_{\rm WZ}[g_s]}{\di s} = \dfrac{1}{4\pi} \int_{\Sigma} \Tr \Big\{ g_s^{-1} \partial_s  g_s \, (g_s^{-1} \di g_s)^2 \Big\}
\end{equation}
The above formula establishes the required rate of change of the Wess--Zumino action with respect to homotopic changes in the field. In particular, it manifestly shows that the derivative of the Wess--Zumino action is independent of the choice of $\widetilde \Sigma$ and of a smooth family  $\widetilde g_s$ of extension%
\footnote{From the field theory point of view, it means that the ambiguity appearing in the definition of the Wess--Zumino action vanishes in the corresponding equations of motion.}.

\begin{remark}[Variation of $\WZ$ under homotopy]
From the previous formula we deduce that for two homotopic fields $g_0$ and $g_1$ one has
\[\WZ[g_1]-\WZ[g_0] = \dfrac{1}{12\pi} \int_{[0,1]\times \Sigma} \Tr \left\{ (\widehat g^{-1} \di \widehat g)^3\right\}, \]
where $\widehat g(s,\sigma) := g_s(\sigma)$, $(s,\sigma) \in [0,1] \times \Sigma$. This formula has to be understood modulo $2 \pi \mathbb Z$ in general as it depends on the choice of the homotopy. The above identity can be applied to understand the definition \eqref{eq:def_K} of $\mathcal K$. Notice that in that case $\Phi$ provides an homotopy between $\phi_0$ and $\phi_\pi$, the homotopy parameter being $s=k_1 \in [0,\pi]$. Thus we have in that case
\[\WZ[\phi_\pi]-\WZ[\phi_0] = -\dfrac{1}{12\pi} \int_{S^1 \times \EBZ} \Tr \left\{ (\Phi^{-1} \di \Phi)^3\right\}, \]
the minus sign coming from the fact that the homotopy parameter is in second position in $S^1 \times \EBZ$. This equality provides a direct proof that $\mathcal K^2 = 1$, and  suggest an interpretation of $\mathcal K$ as an obstruction to the validity a time-reversal equivariant version of the previous equality (indeed the homotopy $\Phi$ is not equivariant in the sense of \eqref{EquivHomotopy}, as $s$ is sent to $-s$ under time-reversal symmetry).
\end{remark}

\subsection{Polyakov--Wiegmann formula}
 
As mentioned above, the Polyakov--Wiegmann formula \cite{PolyakovWiegmann} is used to compute the Wess--Zumino amplitude for the product of two fields $g, h \colon \Sigma \to G$ (defined pointwise as $gh(\sigma) = g(\sigma) \, h(\sigma)$) when $G$ is compact and simply connected, and has been generalized to any compact simple Lie group \cite{GawedzkiWaldorf09} (compare also Remark~\ref{rk:PWanomaly} below). Aiming at applications where $\Sigma = \T^2$ and $G = U(N)$ is neither simple nor simply connected, we investigate this setting by making use of homotopic deformations of the fields.

We start from a general result. Given two fields $g, h \colon \Sigma \to G$, we define the \emph{Polyakov--Wiegmann functional}
\[ \mathrm{PW}[g,h] := S_{\rm WZ}[g h] - S_{\rm WZ}[g]- S_{\rm WZ}[h] -\dfrac{1}{4\pi} \int_{\Sigma} (g \times h)^*\alpha, \]
where
\[ (g \times h)^*\alpha := - \Tr(g^{-1} \di g\,  \di h \, h^{-1}) \]
(\ie $(g \times h)^*\alpha$ is the pullback via the map $g \times h \colon \Sigma \to G \times G$, $(g \times h)(\sigma) = (g(\sigma), h(\sigma))$, of the differential $2$-form $\alpha$ on $G \times G$ defined by the right-hand side of the above equality).

\begin{proposition}\label{prop:PWh}
Let $g_0, \, g_1 \colon \Sigma \to G$ and $h_0, \, h_1 \colon \Sigma \to G$ be two pairs of homotopic fields. Then
\[ \mathrm{PW}[g_0,h_0] = \mathrm{PW}[g_1,h_1]. \]
\end{proposition}
\begin{proof}
We prove that
\begin{equation}\label{dPW}
\dfrac{\di \,}{\di s} \mathrm{PW}[g_s,h_s] = 0.
\end{equation}
The proof of \eqref{dPW} just requires formula \eqref{dSwz} and some differential calculus. In order to have a lighter notation, we drop from here on after the dependence on $s$, and denote $g'$ for $\partial_s g_s$ (and $\di$ is still the total derivative on the $\Sigma$-variable). In this notation
\begin{equation}
\dfrac{\di \,}{\di s} S_{\rm WZ}[g h] = \dfrac{1}{4\pi} \int_{\Sigma} \Tr \Big\{ h^{-1}g^{-1} (gh)' \,   (h^{-1} g^{-1} \di (gh) )^2 \Big\}
\end{equation}
Expanding both kind of derivatives and using the cyclicity of the trace
\begin{align*}
 &\Tr \Big\{  h^{-1}g^{-1} (gh)'   (h^{-1} g^{-1} \di(gh) )^2 \Big\} \cr
 & =  \Tr \Big\{  (h^{-1}g^{-1} g' h  + h^{-1} h')   (h^{-1} g^{-1} \di g\, h + h^{-1} \di h)^2 \Big\}\cr
 & = \Tr \Big\{ (h^{-1}g^{-1} g' h  + h^{-1} h') h^{-1}  \big( (g^{-1} \di g)^2  + g^{-1} \di g \di h\, h^{-1} + \di h\, h^{-1} g^{-1} \di g +  (\di h\, h^{-1})^2  \big)h \Big\}\cr
 & = \Tr \Big\{ g^{-1} g' (g^{-1} \di g)^2 + h^{-1} h' (h^{-1} \di h)^2 + g^{-1} g'  g^{-1} \di g  \di h \,h^{-1} + g^{-1} g' \di h \,h^{-1} \, g^{-1} \di g \cr
 & \hspace{1cm} + g^{-1} g' (\di h \,h^{-1})^2+ h' h^{-1}(g^{-1} \di g)^2 + h' h^{-1}  g^{-1} g' \di h \,h^{-1} + h' h^{-1} \di h \,h^{-1} \, g^{-1} \di g\Big\}.
\end{align*}
The first two terms on the right-hand side correspond to the derivative of the Wess--Zumino action for the fields $g$ and $h$, see \eqref{dSwz}, so that
\begin{equation}\label{inter1}
\begin{aligned}
 \dfrac{\di \,}{\di s} \Big( S_{\rm WZ}[gh] & - S_{\rm WZ}[g]- S_{\rm WZ}[h] \Big) \\ 
 & = \dfrac{1}{4\pi} \int_\Sigma \Tr \Big\{ g^{-1} g' \big(g^{-1} \di g  \di h \,h^{-1} +  \di h \,h^{-1} g^{-1} \di g  + (\di h \,h^{-1})^2\big)\cr
 & \hspace{2cm} +  h' h^{-1} \big(g^{-1} \di g  \di h \,h^{-1} +  \di h \,h^{-1} g^{-1} \di g  + (g^{-1} \di g )^2\big) \Big\}.
\end{aligned}
\end{equation}

Finally,
\[ \dfrac{\di \,}{\di s} \Big( -\dfrac{1}{4\pi} \int_{\Sigma} (g,h)^*\alpha \Big) = \dfrac{1}{4\pi} \int_{\Sigma} \Tr \big\{ \big( g^{-1} \di g \di h \, h^{-1} \big)' \big\} \]
where
\begin{align*}
\Tr\Big\{ \big( g^{-1} \di g \di h \, h^{-1} \big)' \Big\}  = \Tr\Big\{ &- g^{-1} g' g^{-1} \di g \di h\, h^{-1}+ g^{-1} \di(g') \di h\, h^{-1} \\
& + g^{-1} \di g(\di h') \, h^{-1} - g^{-1} \di g  \di h \, h^{-1} h' h^{-1} \Big\}.
\end{align*}
The first and last terms already cancel the first and fourth ones in \eqref{inter1}, whereas
\[ \Tr\big\{ g^{-1} \di(g')\di h\, h^{-1} \big \}= \di \Tr\{ g^{-1} g' \di h\, h^{-1}\} - \Tr \big\{ -g^{-1} \di g\, g^{-1} g' \di h\, h^{-1} + g^{-1} g' (\di h\, h^{-1})^2\big\}. \]
When integrated, the first term is zero using Stokes formula, since $\partial \Sigma = \emptyset$, and the remaining ones cancel with the second and third of \eqref{inter1} since $ \Tr \{ -g^{-1} \di g\, g^{-1} g' \di h\, h^{-1}\} =  \Tr \{ g^{-1} g' \di h\, h^{-1} g^{-1} \di g\}$. With a similar argument of integration by parts on $\Tr\{g^{-1} \di g(\di h') \, h^{-1}\}$, we cancel the last two remaining terms in \eqref{inter1} and conclude the proof of \eqref{dPW}.
\end{proof}

\subsection{Adjoint Polyakov--Wiegmann formula}

We now consider a version of the Polyakov--Wiegmann functional when the product of fields is replaced by the adjoint action, namely we focus on the \emph{adjoint Polyakov--Wiegmann functional}
\[ \mathrm{APW}[g,h] := S_{\rm WZ}[g h  g^{-1}] - S_{\rm WZ}[h] -\dfrac{1}{4\pi} \int_{\Sigma}  (g \times h )^*\beta \]
where
\begin{equation}\label{def_beta}
(g \times h)^*\beta = -  \Tr \Big\{ h (g^{-1} \di g) h^{-1} (g^{-1} \di g) + g^{-1} \di g(h^{-1} \di h + \di h\, h^{-1}) \Big\}.
\end{equation}

\begin{proposition} \label{prop:APW01}
Let $g_0, \, g_1 \colon \Sigma \to G$ and $h_0, \, h_1 \colon \Sigma \to G$ be two pairs of homotopic fields. Then
\[ \mathrm{APW}[g_0,h_0] = \mathrm{APW}[g_1,h_1]. \]
\end{proposition}
\begin{proof}
We denote by $g_s$ and $h_s$ two smooth homotopies between the fields, but we drop the dependence on $s$ in what follows to streamline the notation. 

Replacing $h$ with $h g^{-1}$ in \eqref{dPW} we get
\[ \dfrac{\di \,}{\di s} \Big( S_{\rm WZ}[g h  g^{-1}] - S_{\rm WZ}[g]- S_{\rm WZ}[h g^{-1}] -\dfrac{1}{4\pi} \int_{\Sigma} (g \times h g^{-1})^*\alpha \Big) = 0. \]
Besides,
\[ \dfrac{\di \,}{\di s} \Big( S_{\rm WZ}[ h g^{-1} ] - S_{\rm WZ}[h]- S_{\rm WZ}[g^{-1}] -\dfrac{1}{4\pi} \int_{\Sigma} (h \times g^{-1})^*\alpha \Big) = 0 
\]
and $S_{\rm WZ}[g^{-1}] = - S_{\rm WZ}[g]$ as follows from $\Tr\{ (\widetilde g \, \di (\widetilde g^{-1}))^3 \} = - \Tr\{ (\widetilde g^{-1} \di \widetilde g)^3\}$ and Definition~\ref{def:WZ}. Summing the two previous equations we obtain
\[ \dfrac{\di \,}{\di s} \Big( S_{\rm WZ}[g h  g^{-1}] - S_{\rm WZ}[h] -\dfrac{1}{4\pi} \int_{\Sigma}  (g \times h g^{-1})^*\alpha +  (h \times g^{-1})^*\alpha \Big)  = 0. \]
Finally
\begin{align*}
(g \times h g^{-1})^*\alpha +  (h \times g^{-1})^*\alpha & =  -\Tr\big\{ g^{-1} \di g \di(h g^{-1}) g h^{-1} \big\} - \Tr\big\{ h^{-1}\di h (\di g^{-1}) g \big\} \cr
&= -\Tr\big\{ g^{-1} \di g \di h \, h^{-1} \big\}  + \Tr\big\{ g^{-1} \di g h g^{-1} \di g\,  h^{-1} \big\} + \Tr\big\{ h^{-1}\di h g^{-1} \di g \big\} \cr
&= - \Tr\Big\{ g^{-1} \di g \di h \, h^{-1} +  h g^{-1} \di g\,  h^{-1} g^{-1} \di g + g^{-1} \di g h^{-1}\di h \Big\}\cr
&= (g \times h)^*\beta,
\end{align*}
compare \eqref{def_beta}. We deduce then
\begin{equation}\label{dCPW}
\dfrac{\di \,}{\di s} \mathrm{APW}[g_s,h_s]  = 0
\end{equation}
which concludes the proof.
\end{proof}

\subsection{\texorpdfstring{Equivariant adjoint Polyakov--Wiegmann formula for $U(N)$-valued fields on the torus}{Equivariant adjoint Polyakov--Wiegmann formula for U(N)-valued fields on the torus}}

Using the ``normal form'' for the (equivariant) homotopy class of a map $g \colon \T^2 \to U(N)$, provided by Lemma~\ref{T2->U(N)} in the Appendix, we are able to prove the adjoint Polyakov--Wiegmann formula and its equivariant version for $U(N)$-valued fields defined on $\Sigma = \T^2$.

\begin{theorem}[(Equivariant) adjoint Polyakov--Wiegmann formula] 
\label{thm:APW}
Let $g, h \colon \T^2 \to U(N)$ be two fields. Then
\begin{equation} \label{eqn:APW}
\exp \left(\iu \, \WZ[g h g^{-1}]\right) = \exp \left(\iu \, \WZ[h]\right) \, \exp \left(\frac{\iu}{4 \pi} \int_{\Sigma} (g \times h)^* \beta \right)
\end{equation}
where $(g \times h)^* \beta$ is as in \eqref{def_beta}. 

If moreover the two fields are equivariant, then
\begin{equation} \label{eqn:EAPW}
\sqrt{\exp \left(\iu \, \WZ[g h g^{-1}]\right)} = \sqrt{\exp \left(\iu \, \WZ[h]\right)} \, \exp \left(\frac{\iu}{8 \pi} \int_{\Sigma} (g \times h)^* \beta \right).
\end{equation}
\end{theorem}
\begin{proof}
We begin with the non-equivariant case. Each field $g \colon \T^2 \to U(N)$ is characterized up to homotopy by the two winding numbers $(n_g, m_g) \in \Z^2$ along the two independent loops in $\T^2 = \T \times \T$, by virtue of part~\ref{2d-NonEquiv} of Lemma~\ref{T2->U(N)}. In particular, $g$ is homotopic to
\begin{equation}\label{eqn:Normalform_g1}
 g_1(k_1, k_2) := \diag \left( \eu^{\iu (k_1 n_g + k_2 m_g)}, 1, \ldots, 1 \right), \quad (k_1, k_2) \in \T^2. 
\end{equation}
Similarly, $h$ is homotopic to $h_1$ of the same form. Since $g_1 h_1 g_1^{-1} = h_1$, one can readily compute
\begin{equation} \label{APW1}
\mathrm{APW}[g_1, h_1] = - \frac{1}{4 \pi} \int_{\Sigma} (g_1 \times h_1)^* \beta = - 2 \pi \left(n_g m_h - m_g n_h \right) \in 2 \pi \Z.
\end{equation}
In view of Proposition~\eqref{prop:APW01}, we have that also $\mathrm{APW}[g,h] \in 2 \pi \Z$, and consequently $\eu^{\iu \, \mathrm{APW}[g,h]} = 1$. Spelling out this equality gives exactly \eqref{eqn:APW}.

A similar argument holds in the equivariant case. In this setting, the ``normal form'' of the field $g$ prescribed by part~\ref{2d-YesEquiv} of Lemma~\ref{T2->U(N)} is
\[ g_1(k_1, k_2) := \diag \left( \eu^{\iu (k_1 n_g + k_2 m_g)}, \eu^{\iu (k_1 n_g + k_2 m_g)}, 1, \ldots, 1 \right), \quad (k_1, k_2) \in \T^2, \]
where $2 n_g \in 2 \Z$ and $2 m_g \in 2 \Z$ are the winding numbers of $\det g$ along the two loops in $\T^2 = \T \times \T$. The field $h$ admits an analogous normal form $h_1$, and again $g_1 h_1 g_1^{-1} = h_1$. A similar computation to \eqref{APW1} shows that this time
\[ \mathrm{APW}[g_1, h_1] = - \frac{1}{4 \pi} \int_{\Sigma} (g_1 \times h_1)^* 
\beta = - 4 \pi \left(n_g m_h - m_g n_h \right) \in 4 \pi \Z, \]
so that $\mathrm{APW}[g, h] \in 4 \pi \Z$ as well, again in view of Proposition~\eqref{prop:APW01}. Consequently $\sqrt{\eu^{\iu \, \mathrm{APW}[g,h]}}$ is well-defined and equals $1$. This readily implies \eqref{eqn:EAPW}.
\end{proof}

\begin{remark}[Anomaly of the Polyakov--Wiegmann formula]\label{rk:PWanomaly}
The usual Polyakov--Wiegmann formula \cite{PolyakovWiegmann} can be compactly written as $\eu^{\iu \mathrm{PW}[g,h]} = 1$, or 
\[ S_{\rm WZ}[g h] = S_{\rm WZ}[g]+ S_{\rm WZ}[h] +\dfrac{1}{4\pi} \int_{\Sigma} (g \times h)^*\alpha \mod 2 \pi \Z. \]
This formula holds for some compact simple Lie groups under certain cohomological conditions \cite{GawedzkiWaldorf09}, but it may fail for other Lie groups, in the sense that for two given fields $g,h: \Sigma \to G$ then $\mathrm{PW}[g,h]$ is not \textit{a priori} in $2 \pi \mathbb Z$. Consequently, the Wess--Zumino amplitude of the product $gh$ is not 
simply related to the ones of $g$ and $h$ as in the previous Theorem for the adjoint case. For example, in the case where $\Sigma = \T^2$ and $G=U(N)$ we can appeal to Proposition~\ref{prop:PWh} to compute the Polyakov--Wiegmann functional for two fields via their normal forms in \eqref{eqn:Normalform_g1}. We end up with
\[\mathrm{PW}[g,h] = - \pi(m_g n_h -n_g m_h)   \]
which is not in $2\pi \mathbb Z$ unless the above combination of winding numbers is even. 

Such an obstruction, or \emph{anomaly}, was already studied in detail for every closed compact $\Sigma$ and every compact simple Lie group in \cite{GawedzkiWaldorf09}, and Theorem~\ref{thm:APW} above states that the adjoint Polyakov--Wiegmann formula has no anomaly for $\Sigma = \mathbb T^2$ and $G=U(N)$. More generally, a detailed classification for simple Lie groups in the context of gauged Wess--Zumino--Witten models shows that the corresponding adjoint version can also be anomalous in some cases \cite{deFromontGawedzkiTauber15}.
\end{remark}

\section{Proof of the main result} \label{sec:Proofs}

We come back to the motivating issue, namely to the proof of Theorem~\ref{thm:Main}. We adopt the streamlined notation introduced in Section~\ref{sec:Factor}; with this, Theorem~\ref{thm:Main} can be restated as

\begin{theorem} \label{thm:WZ=Berry}
Assume that $P(k)$, $k \in \R$, is a smooth and periodic family of projectors. Let $\phi \colon \T^2 \to U(N)$ be the field defined in \eqref{eq:phi}, and let $\A$ be the Berry connection associated to $P(k)$. Then
\begin{equation} \label{eqn:WZ=Berry}
\exp\left(\iu \, \WZ[\phi] \right) = \exp \left( - \iu \oint_{\T} \A \right).
\end{equation}

If moreover $P(k)$ is time-reversal symmetric, then
\begin{equation} \label{eqn:EquivariantWZ=Berry}
\sqrt{\exp\left(\iu \, \WZ[\phi] \right)} = \sqrt{\exp \left( - \iu \oint_{\T} \A \right)}.
\end{equation}
\end{theorem}
\begin{proof}
We compute the two sides of \eqref{eqn:WZ=Berry} and \eqref{eqn:EquivariantWZ=Berry} independently to show that they coincide. 

We start from the Wess--Zumino amplitude of the field $\phi$. First of all, we notice that the field $\psi \colon \T^2 \to U(N)$ has a well-defined Wess--Zumino action, defined according to Definition~\ref{def:WZ}. Indeed, since $\psi \colon S^1 \times \T \to U(N)$ is actually independent of $k$, it can be extended trivially to the solid torus $\widetilde{\Sigma} := S^1 \times \D$, where $\D = \set{ z = r \, \eu^{\iu k} : r \in [0,1], \: k \in \T }$, by setting $\widetilde{\psi}(t,z) := \psi(t)$. Using this extension, we see that
\begin{equation} \label{eqn:WZphi0} 
\WZ[\psi] = \int_{\widetilde{\Sigma}} \widetilde{\psi}^* \, \chi = 0 \mod 2 \pi \Z,
\end{equation}
since the integral vanishes for dimensional reasons. Now we notice that $\phi(t,k)$ is in the adjoint form \eqref{eqn:conjugate}, so that we can appeal to the adjoint Polyakov--Wiegmann identity \eqref{eqn:APW} to compute its Wess--Zumino action as
\[
\WZ[\phi] = \WZ[W \, \psi \, W^{-1}] = \WZ[\psi] + \frac{1}{4 \pi} \int_{S^1 
\times \T} (\psi \times W)^* \beta \mod 2 \pi \Z,
\]
with $(\psi \times W)^* \beta$ as in \eqref{def_beta}. In view of \eqref{eqn:WZphi0}, the above simplifies to
\begin{align*}
\WZ[\phi] & = - \dfrac{1}{4\pi}\int_{S^1 \times \T} \Tr \left\{ \psi(t) \left( W(k)^{-1} \partial_k W(k) \, \di k \right) \psi(t)^{-1} \left( W(k)^{-1} \partial_k W(k) \di k \right) \right\} \\
& \quad - \dfrac{1}{4\pi}\int_{S^1 \times \T} \Tr \left\{ \left( W(k)^{-1} \partial_k W(k) \, \di k \right) \left( \psi(t)^{-1} \partial_t \psi(t) + \partial_t \psi(t) \, \psi(t)^{-1} \right) \di t \right\} \\
& \quad \mod 2 \pi \Z.
\end{align*}
Again by a dimensional argument, the first summand on the right-hand side of the above equality drops. Upon noticing that
\[ \partial_t \psi(t) = 2 \pi \iu \, P(0) \, \psi(t) = 2 \pi \iu \, \psi(t) \, P(0) \]
we are left with
\begin{equation} \label{eqn:WZphibis}
\WZ[\phi] = \iu \int_\T \Tr \left\{ P(0) \, W(k)^{-1} \partial_k W(k) \right\} \, \di k \mod 2 \pi \Z.
\end{equation}

Next, we compute the Berry phase on the right-hand side of \eqref{eqn:WZ=Berry}. Let $\set{e_a(0)}_{1 \le a \le m}$ be any orthonormal basis in $\Ran P(0) \simeq \C^m$. If $\set{P(k)}_{k \in \R}$ satisfies also time-reversal symmetry, we further require that $\set{e_a(0)}_{1 \le a \le m}$ is a symplectic basis for the restriction to $\Ran P(0)$ of the form defined in \eqref{eqn:SymplecticForm}. In view of \eqref{eqn:trivial}, setting
\begin{equation} \label{eqn:frame}
e_a(k) := W(k) \, e_a(0)
\end{equation}
defines an orthonormal basis $\set{e_a(k)}_{1 \le a \le m}$ of $\Ran P(k)$, which is moreover smooth and periodic (and possibly time-reversal symmetric) because so is $W(k)$. Computing the Berry connection as in \eqref{eqn:Berry} with respect to this Bloch frame yields
\[ \A = - \iu \sum_{a=1}^{m} \scal{e_a(0)}{\left( W(k)^* \partial_k W(k) \right) e_a(0)} \di k = - \iu \, \Tr \left\{ P(0) \, W(k)^{-1} \partial_k W(k) \right\} \, \di k \]
owing to the unitarity of $W(k)$. Integrating both sides of the above equality on $\T$ and comparing with \eqref{eqn:WZphibis} we obtain exactly \eqref{eqn:WZ=Berry}. 

The proof of \eqref{eqn:EquivariantWZ=Berry} goes along the same lines, using this time the equivariant adjoint Polyakov--Wiegmann formula \eqref{eqn:EAPW}. The rest of the computation stays unchanged, with the only difference that that all relevant objects are defined mod $4\pi \Z$ rather than $2 \pi \Z$, so that square roots are well-defined.
\end{proof}

\section{Conclusions and perspectives} \label{sec:Conclusions}

Theorem \ref{thm:WZ=Berry} links two $\mathbb Z_2$ invariants and brings along its proof various geometrical objects that suggest several connections with other approaches and possible generalizations, both for physics and mathematics.

First it establishes the equality in Theorem \ref{thm:FKM} and thus provides a direct connection between two geometric approaches that have been developed independently to compute the Fu--Kane--Mele $\mathbb Z_2$ invariant. One, given by $\delta$ from equation \eqref{eq:def_delta}, is based on the Bloch frames associated to the family of projectors $P(\bf{k})$, and the corresponding Berry connection \cite{FiorenzaMonacoPanati16,CorneanMonacoTeufel16}. The other one, given by $\mathcal K$ from equation \eqref{eq:def_K}, is based on the square root of Wess--Zumino amplitudes computed for unitary families associated to $P(\bf{k})$ \cite{Gawedzki15,Lyon15}. 

Even if the two invariants were already matched through the original Pfaffian formula for $\FKM$ \cite{FuKane06} and the gerbe formalism which allows to show that $\mathcal K$ agrees with \eqref{eq:WZ_sr}, we proved that the two invariants are actually equal without referring to these aspects. Instead of localizing the formulas on the four time-reversal symmetric points, we computed explicit expressions for the Wess--Zumino amplitudes (and their square roots) of maps localized on the loops at the boundary of the effective Brillouin zone, the crucial point being that on such loops the family $P(a,k)$ and the corresponding field $\phi_a$ can be factorized in an adjoint form (compare \eqref{eqn:trivial} and \eqref{eqn:conjugate}) so that the (square root of the) Wess--Zumino amplitude of $\phi_a$ can be computed through the (equivariant) adjoint Polyakov--Wiegmann formula.

Note that the effective Brillouin zone was actually introduced by Moore and Balents in \cite{MooreBalents07}, where they proposed to contract the Hamiltonian map $\bk \mapsto H(\bk)$ living on the cylinder $\EBZ$ to one living on a sphere, where a corresponding Chern number can be defined. By imposing time-reversal invariance along the contraction, they showed that this Chern number was defined modulo 2. In some sense the equivariant field extension from Definition \ref{def:EquivField} gives an explicit realization of such a contraction, but for the unitary operator $\exp(2 \pi \iu t P(\bf k) )$ on $[0,1] \times \mathrm{EBZ}$ rather than the Hamiltonian on EBZ. 

Besides, the computation of the Fu--Kane--Mele invariant in terms of loops was also already investigated for a $2$-band  many-body system \cite{LeeRyu08}, where it was shown that $\FKM$ can be expressed in terms of $SU(2)$-Wilson loops, namely the trace of the path-ordered exponential of the integral of the non-abelian Berry connection:
\[ (-1)^{\FKM} = W[\T_\pi] W[\T_0], \quad \text{with} \quad W[\T] = \tr \left\{ \mathcal{P} \, \exp \left(-\iu \oint_{\T} A \right) \right\}, \]
where $A_{ab} = -\iu \scal{e_a}{\di e_b}$, so that $\A = \tr\{A\}$. The path-ordered exponential is a descriptive notation for the holonomy of the Berry connection along the loop $\T$ \cite{KarpMansouriRno00}, which is nothing but the solution of the Cauchy problem for the parallel transport operator, introduced in Remark~\ref{rmk:parallel}, evaluated at the endpoint of the loop $k = 2 \pi$ \cite[Sec.~9.12]{GockelerSchucker}:
\[ \mathcal{P} \, \exp \left(-\iu \oint_{\T} A \right) = P(0) \, T(2 \pi) \, P(0), \]
where the right-hand side should be interpreted as an $m \times m$ matrix acting on $\Ran P(0) \simeq \C^m$. As was already mentioned in Remark~\ref{BlochBundle}, the Berry phase is the determinant of this holonomy, namely
\[ \det \left( P(0) \, T(2 \pi) \, P(0) \right) = \exp \left( - \iu \oint_{\T} \A \right) = \exp \left( -\iu \oint_\T \tr\{A\} \right). \] 
The difference between the Berry phase and the Wilson loop is thus that ``the trace is taken before the exponential'' in the former. While the Wilson loop approach in \cite{LeeRyu08} appears to be restricted to the minimal case $m=2$ for the rank of the projector, our Theorem \ref{thm:FKM} holds for any $m$. However the comparison between the two approaches might be an interesting direction of investigation.

Finally, the proof of Theorem \ref{thm:WZ=Berry} involves the adjoint Polyakov--Wiegmann formula from Theorem  \ref{thm:APW} concerning Wess--Zumino amplitudes for products of fields, which mostly relies on the homotopy classes of the considered maps, characterized by Lemmas \ref{T->U(N)} and \ref{T2->U(N)}. As it was pointed out in Remark \ref{rk:PWanomaly}, the Polyakov--Wiegmann formula and its adjoint version can be anomalous, so that that the Wess--Zumino amplitude of a product map $gh$ is not easily related to the ones of $g$ and $h$. This part of our work also constitutes a first step towards a classification of anomalies for $U(N)$-valued fields, that generalizes the one for simple Lie groups obtained using gerbe techniques \cite{GawedzkiWaldorf09,deFromontGawedzkiTauber15}, for what concerns the Polyakov--Wiegmann formula, its adjoint version, and beyond.

\appendix

\section{Homotopy classes of (equivariant) fields}

In this Appendix, we collect a number of properties concerning homotopy classes of maps $g \colon \T^2 \to U(N)$. As a preliminary step of independent interest, we will need to compute the (equivariant) homotopy classes of maps $\T \to U(N)$.

Recall that, if $X$ is a topological space endowed with an involution $\vartheta$ (\eg $X= \T$ or $X=\T^2$ with involution $\vartheta(\bk) = - \bk$), then two equivariant maps $f_0, f_1 \colon X \to U(2M)$ are called \emph{equivariantly homotopic} if there exists an homotopy $F \colon X \times [0,1] \to U(2M)$, $F(x,s) \equiv f_s(x)$, such that 
\begin{equation} \label{EquivHomotopy}
\Theta \circ f_s = f_s \circ \vartheta \quad \text{for all} \quad s \in [0,1],
\end{equation}
where $\Theta$ is defined in \eqref{Theta}. The homotopy class of the map $f \colon X \to U(N)$ is denoted by $[f]$, and the set of such homotopy classes will be denoted by $[X,U(N)]$. Analogously, $[f]_{\Z_2}$ denotes the equivariant homotopy class of an equivariant map $f \colon X \to U(2M)$, while the set of equivariant homotopy classes will be denoted by $[X,U(2M)]_{\Z_2}$.

\begin{lemma} \label{T->U(N)}
For a smooth map $f \colon \T \to U(N)$ set
\begin{equation} \label{def_deg}
\deg(\det f) := \frac{1}{2 \pi \iu} \oint_{\T} \Tr \left\{ f^{-1} \di f \right\}.
\end{equation}
\begin{enumerate}
 \item \label{1d-NonEquiv} The map $[f] \mapsto \deg(\det f)$ establishes a bijection%
 \footnote{\label{foot:iso} Actually, this is an isomorphism of groups, but we will not need this fact.}%
 \[ [\T, U(N)] \stackrel{1:1}{\longleftrightarrow} \Z. \]
 \item \label{1d-Equiv} The map $[f]_{\Z_2} \mapsto \deg(\det f)$ establishes a bijection\footnoteref{foot:iso}
 \[ [\T, U(2M)]_{\Z_2} \stackrel{1:1}{\longleftrightarrow} 2\Z. \]
\end{enumerate}
\end{lemma}
\begin{proof}
Part \ref{1d-NonEquiv} of the statement follows from quite standard arguments: we include here a sketch of the proof for the readers' convenience. Notice first of all that the set $[\T,U(N)]$ is nothing but the first homotopy group $\pi_1(U(N))$ of the unitary group. The short exact sequence of groups
\[ 1 \longrightarrow SU(N) \longrightarrow U(N) \xrightarrow{\det} U(1) \longrightarrow 1 \]
induces an isomorphism $\pi_1(U(N)) \simeq \pi_1(U(1))$ since the special unitary group is simply connected \cite[Ch.~8, Sec.~12]{Husemoller94}. On the other hand, the homotopy group $\pi_1(U(1)) = [\T,U(1)]$ is isomorphic to the group of integers $\Z$, the isomorphism being given by the winding number of a map $\varphi \colon \T \to U(1)$; the latter can be computed by the Cauchy integral \cite[\S 13.4(b)]{Dubrovin}
\[ \deg(\varphi) = \frac{1}{2 \pi \iu} \oint_{\T} \varphi^{-1} \di \varphi. \]
When $\varphi = \det f$ with $f \colon \T \to U(N)$, the above formula reduces to \eqref{def_deg} (see \eg \cite[Lemma~2.12]{CorneanMonacoTeufel16}).

We now come to part \ref{1d-Equiv} of the statement. We begin by noticing that the fixed-point set for the adjoint action $\Theta$ of $\Z_2$ on $U(2M)$ defined by \eqref{Theta} is given by the group $Sp(M) = U(2M) \cap Sp(2M, \C)$, the unitary group over the quaternions \cite[Prop.~1.139]{Knapp05}. Indeed, the condition $\Theta(g) = g$ for $g \in U(2M)$ can be rewritten as $g^{\mathsf{T}} \, J \, g = J$, where $J$ is the symplectic matrix \eqref{J} and $^{\mathsf{T}}$ denotes transposition: the latter is exactly the condition for a matrix to be 
symplectic. Notice that matrices in $Sp(M)$ have determinant equal to $1$.

Next we show that the map in part \ref{1d-Equiv} is well-defined and provides a bijection. The set of $\Z_2$-equivariant homotopy classes $[\T, U(2M)]_{\Z_2}$ lies as a subset of the set of ``unconstrained'' (\ie non-equivariant) homotopy classes $[\T, U(2M)] = \pi_1(U(2M)) \simeq \Z$, as we have just shown. Now, if $f \colon \T \to U(2M)$ is $\Z_2$-equivariant, \ie $f(-k) = \Theta(f(k))$, then $f(0)$ and $f(\pi)$ are fixed points with respect to $\Theta$, and hence lie in $Sp(M)$ by the considerations above. In particular, $\det f(0) = \det f(\pi) = 1$, so that the map $\det f$ is already periodic on $\T_+ := [0, \pi]$. Moreover, the values it assumes on $\T_+$ completely determine the map $\det f \colon \T \to U(1)$ as $\det f(-k) = \overline{\det f(k)}$ in view of the equivariance condition. It follows that for an equivariant map $f \colon \T \to U(2M)$
\[ \deg(\det f) = 2 \left( \frac{1}{2 \pi \iu} \oint_{\T_+} \Tr \left\{ f^{-1} \di f \right\} \right) \quad \in 2 \Z. \]
Consequently, the map in part \ref{1d-Equiv} of the statement is well-defined: we need to show that it is injective and surjective. 

For an equivariant map $f \colon \T \to U(2M)$, denote by
\[ \deg_+(\det f) := \frac{1}{2 \pi \iu} \oint_{\T_+} \Tr \left\{ f^{-1} \di f \right\}  \quad \in \Z. \]
Assume first of all that $\deg_+(\det f_0) = \deg_+(\det f_1)$. The map $f_0 \big|_{\T_+} \colon \T_+ \to U(2M)$ need not be periodic, even though its determinant is: however, as was already remarked, the matrices $f_0(0)$ and $f_0(\pi)$ lie in the fixed-point set $Sp(M)$ for the action of $\Theta$ on $U(2M)$. The group $Sp(M)$ is path-connected and simply connected \cite[Prop.~1.136]{Knapp05}, so there exists a contractible loop $\ell_0 \colon [0,1] \to Sp(M)$ such that $\ell_0(0) = f_0(0)$ and $\ell_0(1) = f_0(\pi)$. Notice that by definition of $Sp(M)$ the loop $\ell_0$ is also equivariantly contractible. Consider the concatenation of maps $\widetilde{f}_0 := \ell_0 \sharp \left(f_0 \big|_{\T_+}\right)$: this is now a $U(2M)$-valued periodic map (and we write $\widetilde{f}_0 \colon \widetilde{\T} \to U(2M)$, where $\widetilde{\T}$ is the concatenation of the two intervals $\T_+$ and $[0,1]$ with endpoints identified), so it uniquely determines a class in $\pi_1(U(2M)) \simeq \Z$ via part \ref{1d-NonEquiv}. Since $\det \ell_0 \equiv 1$, this integer coincides exactly with $\deg_+(\det f_0)$. We argue similarly for $f_1$ and end up with $\widetilde{f}_1 \colon \widetilde{\T} \to U(2M)$, completely specified up to homotopy by $\deg_+(\det f_1)$. By assumption the two integers coincide, and hence in view of the isomorphism of part \ref{1d-NonEquiv} the maps $\widetilde{f}_0$ and $\widetilde{f}_1$ are homotopic. Let $\widetilde{f}_s \colon \widetilde{\T} \to U(2M)$ be an homotopy between them, and define
\[ f_s(k) := \begin{cases} 
\widetilde{f}_s(k) & \text{if } k \in \T_+, \\
\theta^{-1} \, \widetilde{f}_s(-k) \, \theta & \text{if } k \in \T \setminus \T_+.
\end{cases} \]
Since the loops $\ell_0$ and $\ell_1$ are equivariantly contractible, the above defines an equivariant homotopy between $f_0$ and $f_1$, and the map $[f]_{\Z_2} \mapsto \deg(\det f)$ is injective.

Finally, to check surjectivity it suffices to notice that 
\[ f_n(k) := \mathrm{diag}(\eu^{\iu n k}, \eu^{\iu n k}, 1, \ldots, 1),  \quad n \in \Z \] 
defines an equivariant map and has $\deg_+(\det f_n) = n$.
\end{proof}

The above result is the main building block in studying the homotopy classes of fields $\T^2 \to U(N)$.

\begin{lemma} \label{T2->U(N)}
For a map $g \colon \T^2 \to U(N)$, denote by $g\sub{L} \colon \T\sub{L} \to U(N)$ (resp. $g\sub{R} \colon \T\sub{R} \to U(N)$) the restriction of $g$ to $\T\sub{L} = \T \times \set{0} \subset \T \times \T = \T^2$ (resp. to $\T\sub{R} = \set{0} \times \T \subset \T \times \T = \T^2$). 
\begin{enumerate}
 \item \label{2d-NonEquiv} The map $[g] \mapsto \left( \deg(\det g\sub{L}), \deg(\det g\sub{R}) \right)$ establishes a bijection
 \[ [\T^2, U(N)] \stackrel{1:1}{\longleftrightarrow} \Z^2. \]
 \item \label{2d-YesEquiv} The map $[g]_{\Z_2} \mapsto \left( \deg(\det g\sub{L}), \deg(\det g\sub{R}) \right)$ establishes a bijection
 \[ [\T^2, U(2M)]_{\Z_2} \stackrel{1:1}{\longleftrightarrow} (2\Z)^2. \]
\end{enumerate}
\end{lemma}
\begin{proof}
In view of Lemma~\ref{T->U(N)}, the restrictions $g\sub{L}$ and $g\sub{R}$ are completely specified up to homotopy by the winding numbers of their determinants, defined as in \eqref{def_deg}. It is a fundamental result in (equivariant) obstruction theory (see \cite[Ch.~7]{DavisKirk01} and \cite[Ch.~2]{Bredon}) that the obstruction to lift an (equivariant) homotopy from the $1$-skeleton $\T\sub{L} \cup \T\sub{R}$ of the $2$-torus to the whole $\T^2$ is encoded in a cohomology class with coefficients in $\pi_2(U(N))$. Since the latter second homotopy group is trivial \cite[Ch.~8, Sec.~12]{Husemoller94}, all homotopies on the $1$-skeleton extend to the $2$-skeleton, and this concludes the proof.
\end{proof}

\begin{remark}
The computation of spaces of (equivariant) homotopy maps from $\T^d$ to $U(N)$ are not new in the literature, see for example \cite{KennedyGuggenheim15} and references therein (where applications to other classes of topological insulators are also discussed). We decided however to include Lemmas~\ref{T->U(N)} and \ref{T2->U(N)} in our presentation because abstract proofs from algebraic topology usually fail to describe explicitly the isomorphisms involved. The characterization of the bijections in Lemmas~\ref{T->U(N)} and \ref{T2->U(N)} in terms of winding numbers has been repeatedly used throughout the paper.
\end{remark}

%%%%% BIBLIOGRAPHY %%%%%

\bigskip \bigskip

\footnotesize

\begin{tabular}{ll}

(D. Monaco) & \textsc{Fachbereich Mathematik, Eberhard Karls Universit\"{a}t T\"{u}bingen} \\
 &  Auf der Morgenstelle 10, 72076 T\"{u}bingen, Germany \\
 &  {E-mail address}: \href{mailto:domenico.monaco@uni-tuebingen.de}{\texttt{domenico.monaco@uni-tuebingen.de}} \\
\\
(C. Tauber) & \textsc{Dipartimento di Matematica, ``La Sapienza'' Universit\`{a} di Roma} \\
 &  Piazzale Aldo Moro 2, 00185 Rome, Italy \\
 & \textsl{and} \\
 & \textsc{Institute for Theoretical Physics, ETH Z\"{u}rich} \\
 &  Wolfgang-Pauli-Str. 27, 8093 Z\"{u}rich, Switzerland \\
 &  {E-mail addresses}: \href{mailto:tauber@mat.uniroma1.it}{\texttt{tauber@mat.uniroma1.it}}, \href{mailto:tauberc@itp.phys.ethz.ch}{\texttt{tauberc@itp.phys.ethz.ch}}  \\
\end{tabular}
\end{document}